\documentclass[11pt]{article}
\usepackage[a4paper]{geometry}
\usepackage{amsfonts, amsmath, amssymb, amsthm, graphicx, caption, authblk, multirow, makecell, framed, float, xcolor, enumitem, tikz, hyperref}
\theoremstyle{definition}

\newtheorem{lemma}{Lemma}

\theoremstyle{remark}

\definecolor{blk}{RGB}{63,63,63}
\newcommand*{\mybox}[1]{%
  \framebox{\raisebox{0cm}[0.5\baselineskip][0.05\baselineskip]{%
    \hbox to 0.10cm {\hss#1\hss}}}\hspace{0.05cm}}

\begin{document}
\title{Tatami Printer: Physical ZKPs for Tatami Puzzles}
\author[1]{Suthee Ruangwises\thanks{\texttt{suthee@cp.eng.chula.ac.th}}}
\affil[1]{Department of Computer Engineering, Chulalongkorn University, Bangkok, Thailand}
\date{}
\maketitle

\begin{abstract}
Tatami puzzles are pencil puzzles with an objective to partition a rectangular grid into rectangular regions such that no four regions share a corner point, as well as satisfying other constraints. In this paper, we develop a physical card-based protocol called \textit{Tatami printer} that can help verify solutions of Tatami puzzles. We then use the Tatami printer to construct zero-knowledge proof protocols for two such puzzles: Tatamibari and Square Jam. These protocols enable a prover to show a verifier the existence of the puzzles' solutions without revealing them.

\end{abstract}

\section{Introduction}
Pencil puzzles are puzzles written on paper that can be solved using logical reasoning. Examples of them include Sudoku, Numberlink, and Nonogram. Pencil puzzles are classified into several categories based on their main theme. Puzzles involving partitioning a rectangular grid into multiple regions to satisfy certain rules are often called \textit{decomposition puzzles}.

Many decomposition puzzles, especially the ones originated in Japan, specify that the regions must be rectangular, and also have a common \textit{corner constraint}: no four regions can share a corner point. This rule came from the arrangement of Japanese Tatami mats. Hence, these puzzles can be collectively called \textit{Tatami puzzles}, a subcategory of decomposition puzzles \cite{puzz}.

\subsection{Tatamibari}
\textit{Tatamibari} is a Tatami puzzle developed by a Japanese publisher Nikoli. The puzzle consists of an $m \times n$ rectangular grid, with some cells containing a $+$, $|$, or $-$ symbol. The objective of this puzzle is to partition the grid into rectangles, with each one containing exactly one symbol.
\begin{itemize}
	\item If the symbol is a $+$, that rectangle must be a square.
	\item If the symbol is a $|$, that rectangle must have a greater height than width.
	\item If the symbol is a $-$, that rectangle must have a greater width than height.
\end{itemize}
In addition, no four rectangles can share a corner point \cite{janko2}. See Fig. \ref{fig0}.

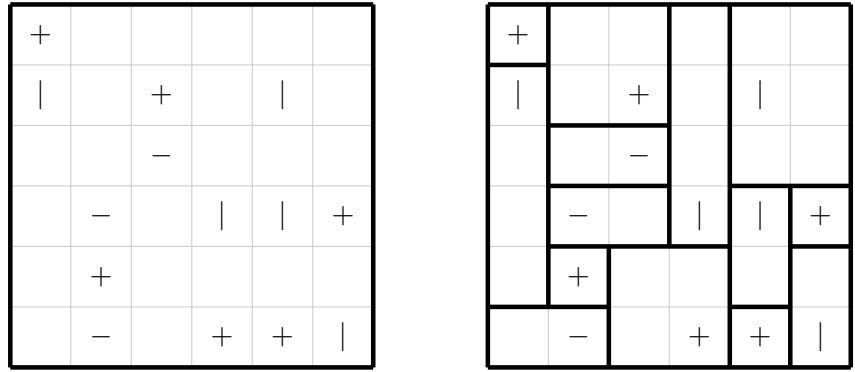
\begin{figure}
\centering
\begin{tikzpicture}
\draw[step=0.8cm,color={rgb:black,1;white,4}] (0,0) grid (4.8,4.8);

\draw[line width=0.6mm] (0,0) -- (4.8,0);
\draw[line width=0.6mm] (0,4.8) -- (4.8,4.8);
\draw[line width=0.6mm] (0,0) -- (0,4.8);
\draw[line width=0.6mm] (4.8,0) -- (4.8,4.8);

\node at (1.2,0.4) {$-$};
\node at (2.8,0.4) {$+$};
\node at (3.6,0.4) {$+$};
\node at (4.4,0.4) {$|$};
\node at (1.2,1.2) {$+$};
\node at (1.2,2) {$-$};
\node at (2.8,2) {$|$};
\node at (3.6,2) {$|$};
\node at (4.4,2) {$+$};
\node at (2,2.8) {$-$};
\node at (0.4,3.6) {$|$};
\node at (2,3.6) {$+$};
\node at (3.6,3.6) {$|$};
\node at (0.4,4.4) {$+$};
\end{tikzpicture}
\hspace{1.2cm}
\begin{tikzpicture}
\draw[step=0.8cm,color={rgb:black,1;white,4}] (0,0) grid (4.8,4.8);

\draw[line width=0.6mm] (0,0) -- (4.8,0);
\draw[line width=0.6mm] (0,4.8) -- (4.8,4.8);
\draw[line width=0.6mm] (0,0) -- (0,4.8);
\draw[line width=0.6mm] (4.8,0) -- (4.8,4.8);

\draw[line width=0.6mm] (0,0.8) -- (1.6,0.8);
\draw[line width=0.6mm] (3.2,0.8) -- (4,0.8);
\draw[line width=0.6mm] (0.8,1.6) -- (3.2,1.6);
\draw[line width=0.6mm] (4,1.6) -- (4.8,1.6);
\draw[line width=0.6mm] (0.8,2.4) -- (2.4,2.4);
\draw[line width=0.6mm] (3.2,2.4) -- (4.8,2.4);
\draw[line width=0.6mm] (0.8,3.2) -- (2.4,3.2);
\draw[line width=0.6mm] (0,4) -- (0.8,4);

\draw[line width=0.6mm] (0.8,0.8) -- (0.8,4.8);
\draw[line width=0.6mm] (1.6,0) -- (1.6,1.6);
\draw[line width=0.6mm] (2.4,1.6) -- (2.4,4.8);
\draw[line width=0.6mm] (3.2,0) -- (3.2,4.8);
\draw[line width=0.6mm] (4,0) -- (4,2.4);

\node at (1.2,0.4) {$-$};
\node at (2.8,0.4) {$+$};
\node at (3.6,0.4) {$+$};
\node at (4.4,0.4) {$|$};
\node at (1.2,1.2) {$+$};
\node at (1.2,2) {$-$};
\node at (2.8,2) {$|$};
\node at (3.6,2) {$|$};
\node at (4.4,2) {$+$};
\node at (2,2.8) {$-$};
\node at (0.4,3.6) {$|$};
\node at (2,3.6) {$+$};
\node at (3.6,3.6) {$|$};
\node at (0.4,4.4) {$+$};
\end{tikzpicture}
\caption{An example of a $6 \times 6$ Tatamibari puzzle (left) and its solution (right)}
\label{fig0}
\end{figure}

Deciding whether a given Tatamibari puzzle has a solution has recently been proved to be NP-complete \cite{np}.

\subsection{Square Jam}
\textit{Square Jam} is another Tatami puzzle developed by Eric Fox. The puzzle consists of an $n \times n$ square grid, with some cells containing a positive integer. The objective of this puzzle is to partition the grid into squares; each square may contain any number (including zero) of cells with an integer, but the integers in those cells must all be equal, and also equal to the side length of that square. In addition, no four squares can share a corner point \cite{janko}. See Fig. \ref{fig00}.

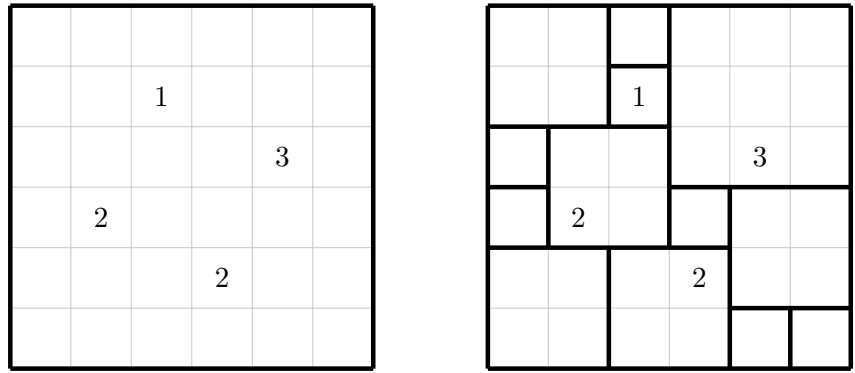
\begin{figure}
\centering
\begin{tikzpicture}
\draw[step=0.8cm,color={rgb:black,1;white,4}] (0,0) grid (4.8,4.8);

\draw[line width=0.6mm] (0,0) -- (4.8,0);
\draw[line width=0.6mm] (0,4.8) -- (4.8,4.8);
\draw[line width=0.6mm] (0,0) -- (0,4.8);
\draw[line width=0.6mm] (4.8,0) -- (4.8,4.8);

\node at (2.8,1.2) {2};
\node at (1.2,2) {2};
\node at (3.6,2.8) {3};
\node at (2,3.6) {1};
\end{tikzpicture}
\hspace{1.2cm}
\begin{tikzpicture}
\draw[step=0.8cm,color={rgb:black,1;white,4}] (0,0) grid (4.8,4.8);

\draw[line width=0.6mm] (0,0) -- (4.8,0);
\draw[line width=0.6mm] (0,4.8) -- (4.8,4.8);
\draw[line width=0.6mm] (0,0) -- (0,4.8);
\draw[line width=0.6mm] (4.8,0) -- (4.8,4.8);

\draw[line width=0.6mm] (3.2,0.8) -- (4.8,0.8);
\draw[line width=0.6mm] (0,1.6) -- (3.2,1.6);
\draw[line width=0.6mm] (0,2.4) -- (0.8,2.4);
\draw[line width=0.6mm] (2.4,2.4) -- (4.8,2.4);
\draw[line width=0.6mm] (0,3.2) -- (2.4,3.2);
\draw[line width=0.6mm] (1.6,4) -- (2.4,4);

\draw[line width=0.6mm] (0.8,1.6) -- (0.8,3.2);
\draw[line width=0.6mm] (1.6,0) -- (1.6,1.6);
\draw[line width=0.6mm] (1.6,3.2) -- (1.6,4.8);
\draw[line width=0.6mm] (2.4,1.6) -- (2.4,4.8);
\draw[line width=0.6mm] (3.2,0) -- (3.2,2.4);
\draw[line width=0.6mm] (4,0) -- (4,0.8);

\node at (2.8,1.2) {2};
\node at (1.2,2) {2};
\node at (3.6,2.8) {3};
\node at (2,3.6) {1};
\end{tikzpicture}
\caption{An example of a $6 \times 6$ Square Jam puzzle (left) and its solution (right)}
\label{fig00}
\end{figure}

\subsection{Zero-Knowledge Proof}
Suppose Agnes designed a difficult Tatami puzzle and challenged her friend Brian to solve it. After a while, he could not solve the puzzle and began to doubt whether it actually has a solution. Agnes needs to convince him that her puzzle has a solution \textit{without} revealing it (which would render the challenge pointless). A \textit{zero-knowledge proof (ZKP)} makes this difficult task possible.

Introduced by Goldwasser et al. \cite{zkp0} in 1989, a ZKP is an interactive protocol between a prover $P$ and a verifier $V$. Both of them are given a computational problem $x$, but its solution $w$ is known to only $P$. A ZKP with perfect completeness and soundness must satisfy the following properties.

\begin{enumerate}
	\item \textbf{Perfect Completeness:} If $P$ knows $w$, then $V$ always accepts.
	\item \textbf{Perfect Soundness:} If $P$ does not know $w$, then $V$ always rejects.
	\item \textbf{Zero-knowledge:} $V$ obtains no information about $w$, i.e. there is a probabilistic polynomial time algorithm $S$ (called a \textit{simulator}), not knowing $w$ but having an access to $V$, such that the outputs of $S$ and of the actual protocol follow exactly the same probability distribution.
\end{enumerate}

In 1991, Goldreich et al. \cite{zkp} proved that a ZKP exists for every NP problem, implying one can construct a computational ZKP for any NP pencil puzzle via a reduction. Such construction, however, is neither practical nor intuitive as it requires cryptographic primitives. Instead, many researchers focused on developing physical ZKPs using objects found in everyday life such as a deck of playing cards. These card-based protocols have the benefit that they do not require computers and also allow external observers to verify that the prover truthfully executes them (which is often a challenging task for digital protocols). They also have didactic values and can be used to teach the concept of ZKP to non-experts.

\subsection{Card-based Zero-Knowledge Proof Protocols}
There is a line of work dedicated to design physical card-based ZKP protocols for specific pencil puzzles, including ABC End View \cite{goishi}, Goishi Hiroi \cite{goishi}, Heyawake \cite{nurikabe}, Hitori \cite{nurikabe}, Juosan \cite{takuzu}, Kakuro \cite{kakuro}, Masyu \cite{slitherlink}, Nonogram \cite{nonogram2}, Numberlink \cite{numberlink}, Nurikabe \cite{nurikabe}, Ripple Effect \cite{ripple}, Shikaku \cite{shikaku}, Slitherlink \cite{slitherlink}, Sudoku \cite{sudoku}, Sumplete \cite{sumplete}, Takuzu \cite{takuzu}, and Toichika \cite{goishi}.

In 2024, Ruangwises and Iwamoto \cite{decom} proposed a more generic protocol called \textit{printing protocol}, which can help verify solutions of several decomposition puzzles such as Five Cells and Meadows. Their protocol, however, has a limitation that it cannot check constraints involving relationships between multiple regions, and thus cannot be used to verify solutions of Tatami puzzles.

\subsection{Our Contribution}
In this paper, we propose a card-based protocol called \textit{Tatami printer}, which is modified from the printing protocol of Ruangwises and Iwamoto \cite{decom}.

The Tatami printer can print numbers or symbols from a template onto a target area from the puzzle grid, constituting a new rectangular region in the grid. It also verifies that the printed region does not overlap with existing regions, and that no four regions share a corner point.

We then use the Tatami printer to construct ZKP protocols for two Tatami puzzles: Tatamibari and Square Jam.

\section{Preliminaries}
\subsection{Cards}
In our protocol, each card may have a number or a symbol written on the front side (e.g. \hbox{\mybox{1},} \mybox{$+$}, \mybox{$\heartsuit$}), or may be a blank card (denoted by \mybox{}). All cards have indistinguishable back sides denoted by \mybox{?}.

\subsection{Pile-Shifting Shuffle}
Given a matrix $M$ of cards, a \textit{pile-shifting shuffle} \cite{polygon} rearranges the columns of $M$ by a uniformly random cyclic shift. It can be implemented by putting all cards in each column into an envelope, and repeatedly picking some envelopes from the bottom and putting them on the top of the pile.

Each card in $M$ can be replaced by a stack of cards (as long as every stack in the same row has the same number of cards), and the protocol still works in the same way.

\subsection{Chosen Cut Protocol} \label{chosen}
Given a sequence $C = (c_1,c_2,...,c_q)$ of $q$ face-down cards, a \textit{chosen cut protocol} \cite{koch} enables $P$ to select a desired card $c_i$ without revealing $i$ to $V$. It also reverts $C$ back to its original state after $P$ finishes using $c_i$.

\begin{figure}
\centering
\begin{tikzpicture}
\node at (0.0,2.4) {\mybox{?}};
\node at (0.6,2.4) {\mybox{?}};
\node at (1.2,2.4) {...};
\node at (1.8,2.4) {\mybox{?}};
\node at (2.4,2.4) {\mybox{?}};
\node at (3.0,2.4) {\mybox{?}};
\node at (3.6,2.4) {...};
\node at (4.2,2.4) {\mybox{?}};

\node at (0.0,2) {$c_1$};
\node at (0.6,2) {$c_2$};
\node at (1.8,2) {$c_{i-1}$};
\node at (2.4,2) {$c_i$};
\node at (3.0,2) {$c_{i+1}$};
\node at (4.2,2) {$c_q$};

\node at (0.0,1.4) {\mybox{?}};
\node at (0.6,1.4) {\mybox{?}};
\node at (1.2,1.4) {...};
\node at (1.8,1.4) {\mybox{?}};
\node at (2.4,1.4) {\mybox{?}};
\node at (3.0,1.4) {\mybox{?}};
\node at (3.6,1.4) {...};
\node at (4.2,1.4) {\mybox{?}};

\node at (0.0,1) {0};
\node at (0.6,1) {0};
\node at (1.8,1) {0};
\node at (2.4,1) {1};
\node at (3.0,1) {0};
\node at (4.2,1) {0};

\node at (0.0,0.4) {\mybox{1}};
\node at (0.6,0.4) {\mybox{0}};
\node at (1.2,0.4) {...};
\node at (1.8,0.4) {\mybox{0}};
\node at (2.4,0.4) {\mybox{0}};
\node at (3.0,0.4) {\mybox{0}};
\node at (3.6,0.4) {...};
\node at (4.2,0.4) {\mybox{0}};
\end{tikzpicture}
\caption{A $3 \times q$ matrix $M$ constructed in Step 1 of the chosen cut protocol}
\label{fig2}
\end{figure}

In the chosen cut protocol, $P$ performs the following steps.
\begin{enumerate}
	\item Construct the following $3 \times q$ matrix $M$ (see Fig. \ref{fig2}).
	\begin{enumerate}
		\item In Row 1, place the sequence $C$.
		\item In Row 2, place a face-down \mybox{1} at Column $i$ and face-down \mybox{0}s at all other columns.
		\item In Row 3, place a face-up \mybox{1} at Column 1 and face-up \mybox{0}s at all other columns.
	\end{enumerate}
	\item Turn all cards face-down. Apply the pile-shifting shuffle to $M$.
	\item Turn over all cards in Row 2 and locate the position of the only \mybox{1}. A card in Row 1 directly above this \mybox{1} will be the desired card $c_i$.
	\item After finishing using $c_i$, place $c_i$ back into $M$ at the same position.
	\item Turn all cards face-down. Apply the pile-shifting shuffle to $M$ again.
	\item Turn over all cards in Row 3 and locate the position of the only \mybox{1}. Shift the columns of $M$ cyclically such that this \mybox{1} moves to Column 1. $M$ is now reverted back to its original state.
\end{enumerate}

Each card in $C$ can be replaced by a stack of cards (as long as every stack has the same number of cards), and the protocol still works in the same way.

\subsection{Printing Protocol of Ruangwises and Iwamoto} \label{print}
In a recent work, Ruangwises and Iwamoto \cite{decom} developed the following printing protocol.

A $p \times q$ matrix of face-down cards (called a \textit{template}) and another $p \times q$ matrix of face-down cards of an area from the puzzle grid are given (the front side of all cards are known to $P$ but not to $V$). This protocol verifies that positions in the area corresponding to non-blank cards in the template are initially empty (consisting of all blank cards). Then, it places all non-blank cards from the template at the corresponding positions in the area, replacing the original blank cards (see Fig. \ref{fig3}) without revealing any card to $V$.

\begin{figure}
\centering
\begin{tikzpicture}
\node at (0,0.6) {\mybox{}};
\node at (0.5,0.6) {\mybox{}};
\node at (1,0.6) {\mybox{}};
\node at (1.5,0.6) {\mybox{}};
\node at (2,0.6) {\mybox{}};

\node at (0,1.2) {\mybox{1}};
\node at (0.5,1.2) {\mybox{2}};
\node at (1,1.2) {\mybox{3}};
\node at (1.5,1.2) {\mybox{}};
\node at (2,1.2) {\mybox{}};

\node at (0,1.8) {\mybox{}};
\node at (0.5,1.8) {\mybox{4}};
\node at (1,1.8) {\mybox{}};
\node at (1.5,1.8) {\mybox{}};
\node at (2,1.8) {\mybox{}};

\node at (1,0.1) {Template};

\node at (2.7,1.2) {\LARGE{+}};
\node at (1,-0.8) {};
\end{tikzpicture}
\begin{tikzpicture}
\node at (0,0.6) {\mybox{}};
\node at (0.5,0.6) {\mybox{1}};
\node at (1,0.6) {\mybox{2}};
\node at (1.5,0.6) {\mybox{3}};
\node at (2,0.6) {\mybox{4}};

\node at (0,1.2) {\mybox{}};
\node at (0.5,1.2) {\mybox{}};
\node at (1,1.2) {\mybox{}};
\node at (1.5,1.2) {\mybox{5}};
\node at (2,1.2) {\mybox{}};

\node at (0,1.8) {\mybox{}};
\node at (0.5,1.8) {\mybox{}};
\node at (1,1.8) {\mybox{}};
\node at (1.5,1.8) {\mybox{}};
\node at (2,1.8) {\mybox{}};

\node at (1,0.1) {Area};

\node at (2.8,1.2) {\LARGE{$\Rightarrow$}};
\node at (1,-0.8) {};
\end{tikzpicture}
\begin{tikzpicture}
\node at (0,0.6) {\mybox{}};
\node at (0.5,0.6) {\mybox{1}};
\node at (1,0.6) {\mybox{2}};
\node at (1.5,0.6) {\mybox{3}};
\node at (2,0.6) {\mybox{4}};

\node at (0,1.2) {\mybox{1}};
\node at (0.5,1.2) {\mybox{2}};
\node at (1,1.2) {\mybox{3}};
\node at (1.5,1.2) {\mybox{5}};
\node at (2,1.2) {\mybox{}};

\node at (0,1.8) {\mybox{}};
\node at (0.5,1.8) {\mybox{4}};
\node at (1,1.8) {\mybox{}};
\node at (1.5,1.8) {\mybox{}};
\node at (2,1.8) {\mybox{}};

\node at (1,0.1) {Area};
\node at (1,-0.8) {};
\end{tikzpicture}

\begin{tikzpicture}
\node at (0,0.6) {\mybox{}};
\node at (0.5,0.6) {\mybox{5}};
\node at (1,0.6) {\mybox{}};
\node at (1.5,0.6) {\mybox{}};
\node at (2,0.6) {\mybox{}};

\node at (0,1.2) {\mybox{1}};
\node at (0.5,1.2) {\mybox{2}};
\node at (1,1.2) {\mybox{3}};
\node at (1.5,1.2) {\mybox{}};
\node at (2,1.2) {\mybox{}};

\node at (0,1.8) {\mybox{}};
\node at (0.5,1.8) {\mybox{4}};
\node at (1,1.8) {\mybox{}};
\node at (1.5,1.8) {\mybox{}};
\node at (2,1.8) {\mybox{}};

\node at (1,0.1) {Template};

\node at (2.7,1.2) {\LARGE{+}};
\end{tikzpicture}
\begin{tikzpicture}
\node at (0,0.6) {\mybox{}};
\node at (0.5,0.6) {\mybox{1}};
\node at (1,0.6) {\mybox{2}};
\node at (1.5,0.6) {\mybox{3}};
\node at (2,0.6) {\mybox{4}};

\node at (0,1.2) {\mybox{}};
\node at (0.5,1.2) {\mybox{}};
\node at (1,1.2) {\mybox{}};
\node at (1.5,1.2) {\mybox{5}};
\node at (2,1.2) {\mybox{}};

\node at (0,1.8) {\mybox{}};
\node at (0.5,1.8) {\mybox{}};
\node at (1,1.8) {\mybox{}};
\node at (1.5,1.8) {\mybox{}};
\node at (2,1.8) {\mybox{}};

\node at (1,0.1) {Area};

\node at (2.8,1.2) {\LARGE{$\Rightarrow$}};
\end{tikzpicture}
\begin{tikzpicture}
\node at (-0.15,-0.05) {};
\node at (2.15,1.95) {};

\node at (1,1.2) {\LARGE{Reject}};
\end{tikzpicture}
\caption{Examples of a successful print from a $3 \times 5$ template onto a $3 \times 5$ area (top) and an unsuccessful print due to overlapping non-blank cards (bottom)}
\label{fig3}
\end{figure}
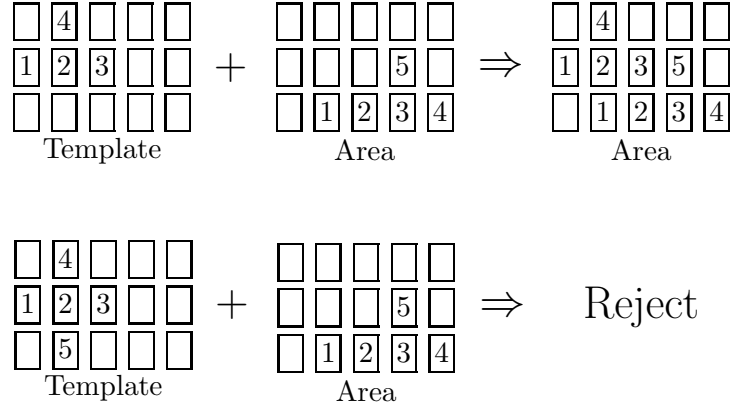

In the printing protocol, $P$ performs the following steps.
\begin{enumerate}
	\item Place each card from the template on top of a corresponding card from the area, creating $pq$ two-card stacks.
	\item For each of the $pq$ stacks, perform the following steps.
	\begin{enumerate}
		\item Apply the chosen cut protocol to select a blank card. (If the preconditions are met, at least one card must be blank; if both cards are blank, $P$ can select any of them.)
		\item Reveal that the selected card is blank (otherwise $V$ rejects) and remove it from the stack.
	\end{enumerate}
\end{enumerate}

After these steps, all non-blank cards from the template are placed at the corresponding positions in the area, and $V$ is convinced that these positions in the area were initially empty.

For decomposition puzzles with simple constraints, $P$ can repeatedly apply the printing protocol to print each region of the partition according to $P$'s solution onto the puzzle grid, convincing $V$ that these regions do not overlap.

However, a major limitation of this protocol is that each region is separately printed, so it cannot check constraints involving relationships between multiple regions. Therefore, it cannot be used to verify solutions of Tatami puzzles.

\section{Our Proposed Protocol: Tatami Printer} \label{tatami}
Considering the limitation of the printing protocol, we aim to modify the protocol so that it can also check the corner constraint of Tatami puzzles.

First, observe that if no four rectangles share a corner point, each point inside (or on the boundary of) the puzzle grid can be a corner point of at most two rectangles. Threfore, on each cell we add a ``corner counter'' to keep track of the number of rectangles having the top-left corner of that cell as their corner point.

On top of each card in a template, we place a \mybox{$\clubsuit$} if its top-left corner is a corner point of a rectangle the template represents, and a \mybox{$\heartsuit$} if its top-left corner is not a corner point. Also, on top of each card in the puzzle grid, we place two cards (initially both \mybox{$\heartsuit$}s) to keep track of the number of rectangles having its top-left corner as their corner point (which is the number of \mybox{$\clubsuit$}s among them).

\begin{figure}
\centering
\begin{tikzpicture}
\node at (0,0.9) {\mybox{$\heartsuit$}};
\node at (0.5,0.9) {\mybox{$\heartsuit$}};
\node at (1,0.9) {\mybox{$\heartsuit$}};
\node at (1.5,0.9) {\mybox{$\heartsuit$}};

\node at (0,1.5) {\mybox{$\clubsuit$}};
\node at (0.5,1.5) {\mybox{$\heartsuit$}};
\node at (1,1.5) {\mybox{$\clubsuit$}};
\node at (1.5,1.5) {\mybox{$\heartsuit$}};

\node at (0,2.1) {\mybox{$\heartsuit$}};
\node at (0.5,2.1) {\mybox{$\heartsuit$}};
\node at (1,2.1) {\mybox{$\heartsuit$}};
\node at (1.5,2.1) {\mybox{$\heartsuit$}};

\node at (0,2.7) {\mybox{$\clubsuit$}};
\node at (0.5,2.7) {\mybox{$\heartsuit$}};
\node at (1,2.7) {\mybox{$\clubsuit$}};
\node at (1.5,2.7) {\mybox{$\heartsuit$}};

\node at (0,0) {\mybox{1}};
\node at (0.5,0) {\mybox{1}};
\node at (1,0) {\mybox{}};
\node at (1.5,0) {\mybox{}};

\node at (0,-0.6) {\mybox{1}};
\node at (0.5,-0.6) {\mybox{1}};
\node at (1,-0.6) {\mybox{}};
\node at (1.5,-0.6) {\mybox{}};

\node at (0,-1.2) {\mybox{}};
\node at (0.5,-1.2) {\mybox{}};
\node at (1,-1.2) {\mybox{}};
\node at (1.5,-1.2) {\mybox{}};

\node at (0,-1.8) {\mybox{}};
\node at (0.5,-1.8) {\mybox{}};
\node at (1,-1.8) {\mybox{}};
\node at (1.5,-1.8) {\mybox{}};

\node at (0.75,-2.3) {Template};

\node at (2.2,0.45) {\LARGE{+}};
\end{tikzpicture}
\begin{tikzpicture}
\node at (0,0.9) {\mybox{$\heartsuit$}\mybox{$\heartsuit$}};
\node at (1,0.9) {\mybox{$\heartsuit$}\mybox{$\heartsuit$}};
\node at (2,0.9) {\mybox{$\clubsuit$}\mybox{$\heartsuit$}};
\node at (3,0.9) {\mybox{$\clubsuit$}\mybox{$\heartsuit$}};

\node at (0,1.5) {\mybox{$\heartsuit$}\mybox{$\heartsuit$}};
\node at (1,1.5) {\mybox{$\heartsuit$}\mybox{$\heartsuit$}};
\node at (2,1.5) {\mybox{$\heartsuit$}\mybox{$\heartsuit$}};
\node at (3,1.5) {\mybox{$\heartsuit$}\mybox{$\heartsuit$}};

\node at (0,2.1) {\mybox{$\heartsuit$}\mybox{$\heartsuit$}};
\node at (1,2.1) {\mybox{$\heartsuit$}\mybox{$\heartsuit$}};
\node at (2,2.1) {\mybox{$\heartsuit$}\mybox{$\heartsuit$}};
\node at (3,2.1) {\mybox{$\heartsuit$}\mybox{$\heartsuit$}};

\node at (0,2.7) {\mybox{$\heartsuit$}\mybox{$\heartsuit$}};
\node at (1,2.7) {\mybox{$\heartsuit$}\mybox{$\heartsuit$}};
\node at (2,2.7) {\mybox{$\clubsuit$}\mybox{$\heartsuit$}};
\node at (3,2.7) {\mybox{$\clubsuit$}\mybox{$\heartsuit$}};

\node at (0,0) {\mybox{}};
\node at (1,0) {\mybox{}};
\node at (2,0) {\mybox{2}};
\node at (3,0) {\mybox{}};

\node at (0,-0.6) {\mybox{}};
\node at (1,-0.6) {\mybox{}};
\node at (2,-0.6) {\mybox{2}};
\node at (3,-0.6) {\mybox{}};

\node at (0,-1.2) {\mybox{}};
\node at (1,-1.2) {\mybox{}};
\node at (2,-1.2) {\mybox{2}};
\node at (3,-1.2) {\mybox{}};

\node at (0,-1.8) {\mybox{}};
\node at (1,-1.8) {\mybox{}};
\node at (2,-1.8) {\mybox{}};
\node at (3,-1.8) {\mybox{}};

\node at (1.5,-2.3) {Area};

\node at (4,0.45) {\LARGE{$\Rightarrow$}};
\end{tikzpicture}
\begin{tikzpicture}
\node at (0,0.9) {\mybox{$\heartsuit$}\mybox{$\heartsuit$}};
\node at (1,0.9) {\mybox{$\heartsuit$}\mybox{$\heartsuit$}};
\node at (2,0.9) {\mybox{$\clubsuit$}\mybox{$\heartsuit$}};
\node at (3,0.9) {\mybox{$\clubsuit$}\mybox{$\heartsuit$}};

\node at (0,1.5) {\mybox{$\clubsuit$}\mybox{$\heartsuit$}};
\node at (1,1.5) {\mybox{$\heartsuit$}\mybox{$\heartsuit$}};
\node at (2,1.5) {\mybox{$\clubsuit$}\mybox{$\heartsuit$}};
\node at (3,1.5) {\mybox{$\heartsuit$}\mybox{$\heartsuit$}};

\node at (0,2.1) {\mybox{$\heartsuit$}\mybox{$\heartsuit$}};
\node at (1,2.1) {\mybox{$\heartsuit$}\mybox{$\heartsuit$}};
\node at (2,2.1) {\mybox{$\heartsuit$}\mybox{$\heartsuit$}};
\node at (3,2.1) {\mybox{$\heartsuit$}\mybox{$\heartsuit$}};

\node at (0,2.7) {\mybox{$\clubsuit$}\mybox{$\heartsuit$}};
\node at (1,2.7) {\mybox{$\heartsuit$}\mybox{$\heartsuit$}};
\node at (2,2.7) {\mybox{$\clubsuit$}\mybox{$\clubsuit$}};
\node at (3,2.7) {\mybox{$\clubsuit$}\mybox{$\heartsuit$}};

\node at (0,0) {\mybox{1}};
\node at (1,0) {\mybox{1}};
\node at (2,0) {\mybox{2}};
\node at (3,0) {\mybox{}};

\node at (0,-0.6) {\mybox{1}};
\node at (1,-0.6) {\mybox{1}};
\node at (2,-0.6) {\mybox{2}};
\node at (3,-0.6) {\mybox{}};

\node at (0,-1.2) {\mybox{}};
\node at (1,-1.2) {\mybox{}};
\node at (2,-1.2) {\mybox{2}};
\node at (3,-1.2) {\mybox{}};

\node at (0,-1.8) {\mybox{}};
\node at (1,-1.8) {\mybox{}};
\node at (2,-1.8) {\mybox{}};
\node at (3,-1.8) {\mybox{}};

\node at (1.5,-2.3) {Area};
\end{tikzpicture}
\caption{An example of a successful print from a $4 \times 4$ template onto a $4 \times 4$ area, with the top matrices being the corner counter part and the bottom matrices being the main part}
\label{fig4}
\end{figure}
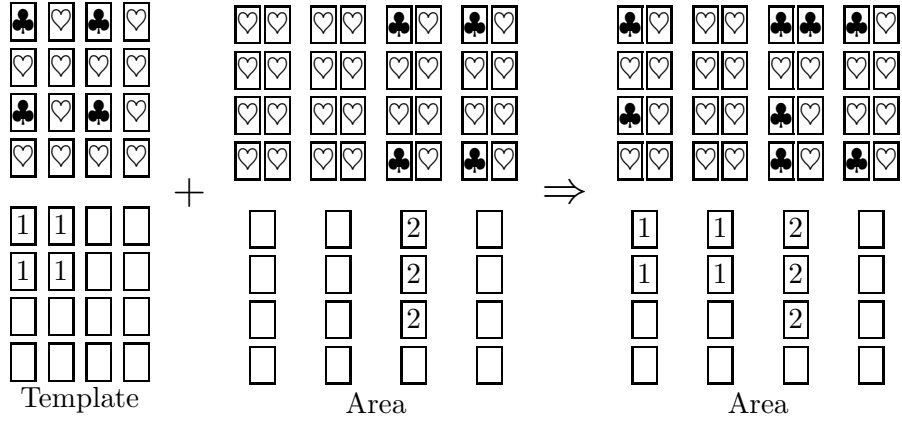

When printing a $p \times q$ template onto a $p \times q$ area from the grid (the front side of all cards are known to $P$ but not to $V$), $P$ performs the following steps.
\begin{enumerate}
	\item In the template, separate the top cards of all stacks (called the \textit{corner counter part}) from the bottom cards (called the \textit{main part}) and put them in another $p \times q$ matrix of cards. Also, in the area, separate the top two cards of all stacks (the corner counter part) from the bottom cards (the main part) and put them in another $p \times q$ matrix of two-card stacks (see Fig. \ref{fig4}).
	\item For the main part, place each card from the template on top of a corresponding card from the area, creating $pq$ two-card stacks.
	\item For each of the $pq$ stacks, perform the following steps.
	\begin{enumerate}
		\item Apply the chosen cut protocol to select a blank card. (If the preconditions are met, at least one card must be blank; if both cards are blank, $P$ can select any of them.)
		\item Reveal that the selected card is blank (otherwise $V$ rejects) and remove it from the stack.
	\end{enumerate}
	\item For the corner counter part, place each card from the template on top of a corresponding two-card stack from the area, creating $pq$ three-card stacks.
	\item For each of the $pq$ stacks, perform the following steps.
	\begin{enumerate}
		\item Apply the chosen cut protocol to select a \mybox{$\heartsuit$}. (If the preconditions are met, at least one card must be a \mybox{$\heartsuit$}; if both cards are \mybox{$\heartsuit$}s, $P$ can select any of them.)
		\item Reveal that the selected card is a \mybox{$\heartsuit$} (otherwise $V$ rejects) and remove it from the stack.
	\end{enumerate}
\end{enumerate}

After these steps, all non-blank cards from the main part of the template are placed at the corresponding positions in the area. All \mybox{$\clubsuit$}s from the corner counter part of the template are also added to the corresponding positions in the area. $V$ is convinced that these positions in the area were initially empty, and that each point in the area is a corner point of at most two rectangles.

\section{ZKP Protocol for Tatamibari} \label{tatamibari}
First, from a solution of Tatamibari, one can fill a symbol on every cell according to the rule of the puzzle (a $+$, $|$, or $-$ depending on the size of a rectangle that cell belongs to). We call this instance an \textit{extended solution} of the puzzle.

The key observation is that in the extended solution of Tatamibari, there are only $mn$ different sizes of rectangles (as there are $m$ possible heights and $n$ possible widths), and the symbol on every cell in a rectangle of each size is fixed. Therefore, only $mn$ different templates are required.

In our protocol, $P$ constructs $mn$ templates, one for each size of rectangle. Each template has size $(m+1) \times (n+1)$ (to support a corner counter of the bottom-right corner), and the rectangle is placed at the top-left corner of the template. A cell inside the rectangle is represented by a card with a symbol according to the rule, while a cell outside the rectangle is represented by a blank card (see Fig. \ref{figA}).

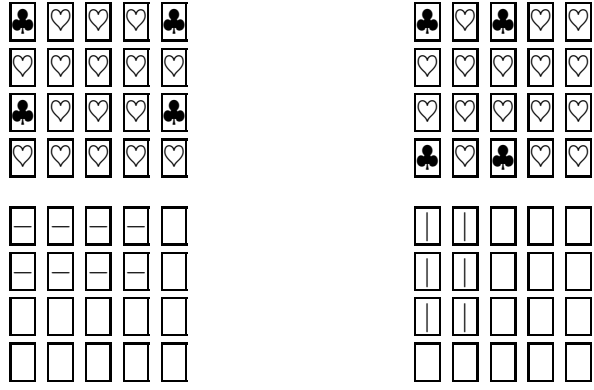
\begin{figure}
\centering
\begin{tikzpicture}
\node at (0,0.9) {\mybox{$\heartsuit$}};
\node at (0.5,0.9) {\mybox{$\heartsuit$}};
\node at (1,0.9) {\mybox{$\heartsuit$}};
\node at (1.5,0.9) {\mybox{$\heartsuit$}};
\node at (2,0.9) {\mybox{$\heartsuit$}};

\node at (0,1.5) {\mybox{$\clubsuit$}};
\node at (0.5,1.5) {\mybox{$\heartsuit$}};
\node at (1,1.5) {\mybox{$\heartsuit$}};
\node at (1.5,1.5) {\mybox{$\heartsuit$}};
\node at (2,1.5) {\mybox{$\clubsuit$}};

\node at (0,2.1) {\mybox{$\heartsuit$}};
\node at (0.5,2.1) {\mybox{$\heartsuit$}};
\node at (1,2.1) {\mybox{$\heartsuit$}};
\node at (1.5,2.1) {\mybox{$\heartsuit$}};
\node at (2,2.1) {\mybox{$\heartsuit$}};

\node at (0,2.7) {\mybox{$\clubsuit$}};
\node at (0.5,2.7) {\mybox{$\heartsuit$}};
\node at (1,2.7) {\mybox{$\heartsuit$}};
\node at (1.5,2.7) {\mybox{$\heartsuit$}};
\node at (2,2.7) {\mybox{$\clubsuit$}};

\node at (0,0) {\mybox{$-$}};
\node at (0.5,0) {\mybox{$-$}};
\node at (1,0) {\mybox{$-$}};
\node at (1.5,0) {\mybox{$-$}};
\node at (2,0) {\mybox{}};

\node at (0,-0.6) {\mybox{$-$}};
\node at (0.5,-0.6) {\mybox{$-$}};
\node at (1,-0.6) {\mybox{$-$}};
\node at (1.5,-0.6) {\mybox{$-$}};
\node at (2,-0.6) {\mybox{}};

\node at (0,-1.2) {\mybox{}};
\node at (0.5,-1.2) {\mybox{}};
\node at (1,-1.2) {\mybox{}};
\node at (1.5,-1.2) {\mybox{}};
\node at (2,-1.2) {\mybox{}};

\node at (0,-1.8) {\mybox{}};
\node at (0.5,-1.8) {\mybox{}};
\node at (1,-1.8) {\mybox{}};
\node at (1.5,-1.8) {\mybox{}};
\node at (2,-1.8) {\mybox{}};
\end{tikzpicture}
\hspace{2.5cm}
\begin{tikzpicture}
\node at (0,0.9) {\mybox{$\clubsuit$}};
\node at (0.5,0.9) {\mybox{$\heartsuit$}};
\node at (1,0.9) {\mybox{$\clubsuit$}};
\node at (1.5,0.9) {\mybox{$\heartsuit$}};
\node at (2,0.9) {\mybox{$\heartsuit$}};

\node at (0,1.5) {\mybox{$\heartsuit$}};
\node at (0.5,1.5) {\mybox{$\heartsuit$}};
\node at (1,1.5) {\mybox{$\heartsuit$}};
\node at (1.5,1.5) {\mybox{$\heartsuit$}};
\node at (2,1.5) {\mybox{$\heartsuit$}};

\node at (0,2.1) {\mybox{$\heartsuit$}};
\node at (0.5,2.1) {\mybox{$\heartsuit$}};
\node at (1,2.1) {\mybox{$\heartsuit$}};
\node at (1.5,2.1) {\mybox{$\heartsuit$}};
\node at (2,2.1) {\mybox{$\heartsuit$}};

\node at (0,2.7) {\mybox{$\clubsuit$}};
\node at (0.5,2.7) {\mybox{$\heartsuit$}};
\node at (1,2.7) {\mybox{$\clubsuit$}};
\node at (1.5,2.7) {\mybox{$\heartsuit$}};
\node at (2,2.7) {\mybox{$\heartsuit$}};

\node at (0,0) {\mybox{$|$}};
\node at (0.5,0) {\mybox{$|$}};
\node at (1,0) {\mybox{}};
\node at (1.5,0) {\mybox{}};
\node at (2,0) {\mybox{}};

\node at (0,-0.6) {\mybox{$|$}};
\node at (0.5,-0.6) {\mybox{$|$}};
\node at (1,-0.6) {\mybox{}};
\node at (1.5,-0.6) {\mybox{}};
\node at (2,-0.6) {\mybox{}};

\node at (0,-1.2) {\mybox{$|$}};
\node at (0.5,-1.2) {\mybox{$|$}};
\node at (1,-1.2) {\mybox{}};
\node at (1.5,-1.2) {\mybox{}};
\node at (2,-1.2) {\mybox{}};

\node at (0,-1.8) {\mybox{}};
\node at (0.5,-1.8) {\mybox{}};
\node at (1,-1.8) {\mybox{}};
\node at (1.5,-1.8) {\mybox{}};
\node at (2,-1.8) {\mybox{}};
\end{tikzpicture}
\caption{$4 \times 5$ templates of rectangles with sizes $2 \times 4$ (left) and $3 \times 2$ (right), with the top matrices being the corner counter part and the bottom matrices being the main part}
\label{figA}
\end{figure}

\subsection{Main Protocol}
Initially, $P$ publicly places two \mybox{$\heartsuit$}s on top of a blank card on every cell in the grid. To handle edge cases, $P$ publicly appends $m$ rows and $n$ columns of ``dummy'' stacks to the bottom and to the right of the grid. Then, $P$ turns all cards face-down. We now have an $2m \times 2n$ matrix of three-card stacks.

Observe that if we arrange all $4mn$ stacks in the matrix into a single sequence $A=(a_1,a_2,...,a_{4mn})$, starting at the top-left corner and going from left to right in Row 1, then from left to right in Row 2, and so on, we can locate exactly where the four neighbors of any given card are. Namely, the cards on the neighbor to the left, right, top, and bottom of a cell containing $a_i$ are $a_{i-1}$, $a_{i+1}$, $a_{i-2n}$, and $a_{i+2n}$, respectively. Therefore, we can select any area from the grid by applying the chosen cut protocol to select the top-left corner cell of that area, and the rest will follow as the chosen cut protocol preserves the cyclic order.

$P$ also constructs $mn$ templates of all $mn$ sizes of rectangle and lets $V$ verify that all templates are correct (otherwise $V$ rejects).

Suppose that in $P$'s extended solution, the grid is partitioned into $k$ rectangles $B_1,B_2,$ $...,B_k$. Note that $k$ is public information, which is the number of cells containing a symbol in the original puzzle. For each $i=1,2,...,k$, $P$ performs the following steps.

\begin{enumerate}
	\item Apply the chosen cut protocol to select an $(m+1) \times (n+1)$ area whose top-left corner is the top-left corner of $B_i$.
	\item Apply the chosen cut protocol to select a template of a rectangle with the same size as $B_i$.
	\item Apply the Tatami printer to the selected template and area.
	\item Reconstruct a template that has just been used and replenish the pile of templates with it. Let $V$ verify again that all $mn$ templates are correct (otherwise $V$ rejects). Note that $V$ does not know which template has just been used.
\end{enumerate}

Finally, $P$ reveals all cards from the main part on the cells that contain a symbol in the original puzzle. $V$ verifies that the symbols on the cards match the ones in the cells (otherwise $V$ rejects). $P$ also reveals the main part of all dummy stacks that they are still blank (otherwise $V$ rejects). If all verification steps pass, then $V$ accepts.

This protocol uses $\Theta(m^2n^2)$ cards and $\Theta(kmn)$ shuffles.

\subsection{Proof of Correctness and Security}
We will prove the perfect completeness, perfect soundness, and zero-knowledge properties of this protocol.

\begin{lemma}[Perfect Completeness] \label{lem1}
If $P$ knows a solution of the Tatamibari puzzle, then $V$ always accepts.
\end{lemma}

\begin{proof}
Suppose $P$ knows an extended solution of the puzzle. Consider each $i$-th iteration of the main protocol.

\begin{itemize}
	\item In Step 3, since $B_1,B_2,...,B_k$ form a partition of the puzzle grid, $B_i$ does not overlap with $B_1,B_2,...,B_{i-1}$. Also, from the corner constraint, each point in the grid can be a corner point of at most two rectangles. Thus, the Tatami printer will pass.
	
	\item In Step 4, since $P$ reconstructs a template that has just been used, all $mn$ templates are correct, and thus this step will pass.
\end{itemize}

Therefore, every iteration will pass. After $k$ iterations, all symbols in $P$'s extended solution will be printed on the grid, so all symbols in the original puzzle will match the ones on the corresponding cards.

Hence, we can conclude that $V$ always accepts.
\end{proof}

\begin{lemma}[Perfect Soundness] \label{lem2}
If $P$ does not know a solution of the Tatamibari puzzle, then $V$ always rejects.
\end{lemma}

\begin{proof}
We will prove the contrapositive of this statement. Suppose $V$ accepts, which means the verification passes in all steps. Consider the main protocol.

Since Step 4 passes for every iteration, all $mn$ templates are correct after each iteration (and also at the beginning of the protocol), which implies the symbols printed in every iteration form a shape of a rectangle and follow the rule of the puzzle.

In Step 3, since the Tatami printer passes for every iteration, $B_i$ does not overlap with $B_1,B_2,...,B_{i-1}$ for every $i$. Also, since the final verification passes, the combined area of $B_1,B_2,...,B_k$ must cover the whole puzzle grid, i.e. $B_1,B_2,...,B_k$ form a partition of the grid. Furthermore, each point in the grid can be a corner point of at most two rectangles, so no four rectangles share a corner point.

Since the final verification passes, all symbols in the original puzzle match the ones on the corresponding cards.

Hence, we can conclude that the puzzle grid is partitioned into rectangles according to the rules, which implies $P$ must know a valid solution of the puzzle.
\end{proof}

\begin{lemma}[Zero-Knowledge] \label{lem3}
During the verification, $V$ obtains no information about $P$'s solution.
\end{lemma}

\begin{proof}
To prove the zero-knowledge property, we will show that any interaction between $P$ and $V$ can be simulated by a simulator $S$ that does not know $P$'s solution. It is sufficient to prove that all distributions of cards that are turned face-up can be simulated by $S$.

\begin{itemize}
	\item In Steps 3 and 6 of the chosen cut protocol in Section \ref{chosen}, because of the pile-shifting shuffles, the \mybox{1} has probability $1/q$ to be at each of the $q$ columns. Therefore, these two steps can be simulated by $S$.

	\item In Steps 3(b) and 5(b) of the Tatami printer in Section \ref{tatami}, there is only one deterministic pattern of the cards that are turned face-up. Therefore, these two steps can be simulated by $S$.
	
	\item In Step 4 of the main protocol, there is only one deterministic pattern of the cards that are turned face-up (all correct templates). Therefore, this step can be simulated by $S$.
\end{itemize}

Hence, we can conclude that $V$ obtains no information about $P$'s solution.
\end{proof}

\section{ZKP Protocol for Square Jam} \label{squarejam}
First, from a solution of Square Jam, one can fill an integer on every cell according to the rule of the puzzle (equal to the side length of a square that cell belongs to). We call this instance an \textit{extended solution} of the puzzle.

The key observation is that in the extended solution of Square Jam, there are only $n$ different sizes of square, and the integer on every cell in a square of each size is fixed. Therefore, only $n$ different templates are required.

However, unlike in Tatamibari, the number of squares in $P$'s solution is not public information, so $P$ cannot reveal that number to $V$. Therefore, $P$ has to apply the Tatami printer for $n^2$ times, sometimes printing nothing, as $V$ knows that there are at most $n^2$ squares.

In our protocol, $P$ constructs $n$ templates, one for each size of square. Each template has size $(n+1) \times (n+1)$, and the square is placed at the top-left corner of the template. A cell inside the square is represented by a card with a number equal to the side length of that square, while a cell outside the square is represented by a blank card (see Fig. \ref{figB}). In addition, $P$ constructs a \textit{blank template} of the same size, with the main part consisting of all blank cards and the corner counter part consisting of all \mybox{$\heartsuit$}s.

\begin{figure}
\centering
\begin{tikzpicture}
\node at (0,0.9) {\mybox{$\heartsuit$}};
\node at (0.5,0.9) {\mybox{$\heartsuit$}};
\node at (1,0.9) {\mybox{$\heartsuit$}};
\node at (1.5,0.9) {\mybox{$\heartsuit$}};

\node at (0,1.5) {\mybox{$\clubsuit$}};
\node at (0.5,1.5) {\mybox{$\heartsuit$}};
\node at (1,1.5) {\mybox{$\clubsuit$}};
\node at (1.5,1.5) {\mybox{$\heartsuit$}};

\node at (0,2.1) {\mybox{$\heartsuit$}};
\node at (0.5,2.1) {\mybox{$\heartsuit$}};
\node at (1,2.1) {\mybox{$\heartsuit$}};
\node at (1.5,2.1) {\mybox{$\heartsuit$}};

\node at (0,2.7) {\mybox{$\clubsuit$}};
\node at (0.5,2.7) {\mybox{$\heartsuit$}};
\node at (1,2.7) {\mybox{$\clubsuit$}};
\node at (1.5,2.7) {\mybox{$\heartsuit$}};

\node at (0,0) {\mybox{2}};
\node at (0.5,0) {\mybox{2}};
\node at (1,0) {\mybox{}};
\node at (1.5,0) {\mybox{}};

\node at (0,-0.6) {\mybox{2}};
\node at (0.5,-0.6) {\mybox{2}};
\node at (1,-0.6) {\mybox{}};
\node at (1.5,-0.6) {\mybox{}};

\node at (0,-1.2) {\mybox{}};
\node at (0.5,-1.2) {\mybox{}};
\node at (1,-1.2) {\mybox{}};
\node at (1.5,-1.2) {\mybox{}};

\node at (0,-1.8) {\mybox{}};
\node at (0.5,-1.8) {\mybox{}};
\node at (1,-1.8) {\mybox{}};
\node at (1.5,-1.8) {\mybox{}};
\end{tikzpicture}
\hspace{2.5cm}
\begin{tikzpicture}
\node at (0,0.9) {\mybox{$\clubsuit$}};
\node at (0.5,0.9) {\mybox{$\heartsuit$}};
\node at (1,0.9) {\mybox{$\heartsuit$}};
\node at (1.5,0.9) {\mybox{$\clubsuit$}};

\node at (0,1.5) {\mybox{$\heartsuit$}};
\node at (0.5,1.5) {\mybox{$\heartsuit$}};
\node at (1,1.5) {\mybox{$\heartsuit$}};
\node at (1.5,1.5) {\mybox{$\heartsuit$}};

\node at (0,2.1) {\mybox{$\heartsuit$}};
\node at (0.5,2.1) {\mybox{$\heartsuit$}};
\node at (1,2.1) {\mybox{$\heartsuit$}};
\node at (1.5,2.1) {\mybox{$\heartsuit$}};

\node at (0,2.7) {\mybox{$\clubsuit$}};
\node at (0.5,2.7) {\mybox{$\heartsuit$}};
\node at (1,2.7) {\mybox{$\heartsuit$}};
\node at (1.5,2.7) {\mybox{$\clubsuit$}};

\node at (0,0) {\mybox{3}};
\node at (0.5,0) {\mybox{3}};
\node at (1,0) {\mybox{3}};
\node at (1.5,0) {\mybox{}};

\node at (0,-0.6) {\mybox{3}};
\node at (0.5,-0.6) {\mybox{3}};
\node at (1,-0.6) {\mybox{3}};
\node at (1.5,-0.6) {\mybox{}};

\node at (0,-1.2) {\mybox{3}};
\node at (0.5,-1.2) {\mybox{3}};
\node at (1,-1.2) {\mybox{3}};
\node at (1.5,-1.2) {\mybox{}};

\node at (0,-1.8) {\mybox{}};
\node at (0.5,-1.8) {\mybox{}};
\node at (1,-1.8) {\mybox{}};
\node at (1.5,-1.8) {\mybox{}};
\end{tikzpicture}
\caption{$4 \times 4$ templates of squares with side lengths 2 (left) and 4 (right), with the top matrices being the corner counter part and the bottom matrices being the main part}
\label{figB}
\end{figure}
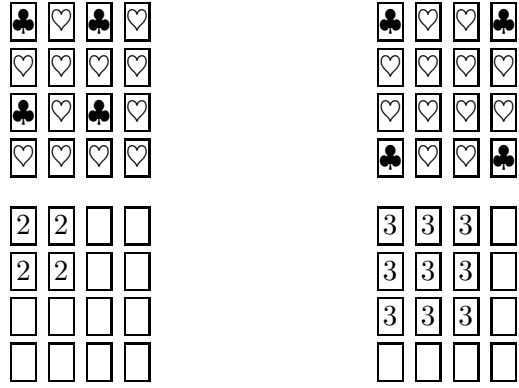

\subsection{Main Protocol}
Initially, $P$ publicly places two \mybox{$\heartsuit$}s on top of a blank card on every cell in the grid. To handle edge cases, $P$ publicly appends $n$ rows and $n$ columns of ``dummy'' stacks to the bottom and to the right of the grid. Then, $P$ turns all cards face-down. We now have a $2n \times 2n$ matrix of three-card stacks.

$P$ also constructs templates of all $n$ sizes of square, plus a blank template, and lets $V$ verify that all $n+1$ templates are correct (otherwise $V$ rejects).

Suppose that in $P$'s extended solution, the grid is partitioned into $k$ squares $B_1,B_2,...,B_k$. For each $i=1,2,...,n^2$, $P$ performs the following steps.

\begin{enumerate}
	\item Apply the chosen cut protocol to select an $(n+1) \times (n+1)$ area whose top-left corner is the top-left corner of $B_i$ (if $i>k$, select any area).
	\item Apply the chosen cut protocol to select a template of a square with the same size as $B_i$ (if $i>k$, select a blank template).
	\item Apply the Tatami printer to the selected template and area.
	\item Reconstruct a template that has just been used and replenish the pile of templates with it. Let $V$ verify again that all $n+1$ templates are correct (otherwise $V$ rejects).
\end{enumerate}

Finally, $P$ reveals all cards from the main part on the cells that contain an integer in the original puzzle. $V$ verifies that the integers on the cards match the ones in the cells (otherwise $V$ rejects). $P$ also reveals the main part of all dummy stacks that they are still blank (otherwise $V$ rejects). If all verification steps pass, then $V$ accepts.

This protocol uses $\Theta(n^3)$ cards and $\Theta(n^4)$ shuffles.

\subsection{Proof of Correctness and Security}
We will prove the perfect completeness, perfect soundness, and zero-knowledge properties of this protocol.

\begin{lemma}[Perfect Completeness] \label{lem4}
If $P$ knows a solution of the Square Jam puzzle, then $V$ always accepts.
\end{lemma}

\begin{proof}
Suppose $P$ knows an extended solution of the puzzle. Consider each $i$-th iteration of the main protocol ($i \leq k$).

\begin{itemize}
	\item In Step 3, since $B_1,B_2,...,B_k$ form a partition of the puzzle grid, $B_i$ does not overlap with $B_1,B_2,...,B_{i-1}$. Also, from the corner constraint, each point in the grid can be a corner point of at most two rectangles. Thus, the Tatami printer will pass.
	
	\item In Step 4, since $P$ reconstructs a template that has just been used, all $n+1$ templates are correct, and thus this step will pass.
\end{itemize}

Therefore, the first $k$ iterations will pass. All subsequent iterations just print a blank template (we call them \textit{trivial iterations}), so they will also pass. After the first $k$ iterations, all integers in $P$'s extended solution will be printed on the grid. Subsequent trivial iterations do not modify the grid, so all integers in the original puzzle will match the ones on the corresponding cards.

Hence, we can conclude that $V$ always accepts.
\end{proof}

\begin{lemma}[Perfect Soundness] \label{lem5}
If $P$ does not know a solution of the Square Jam puzzle, then $V$ always rejects.
\end{lemma}

\begin{proof}
We will prove the contrapositive of this statement. Suppose $V$ accepts, which means the verification passes in all steps. Consider the main protocol.

Since Step 4 passes for every iteration, all $n+1$ templates are correct after each iteration (and also at the beginning of the protocol), which implies the integers printed in every nontrivial iteration form a shape of a square and are equal to the side length of that square.

Let $D_1,D_2,...,D_\ell$ be the squares printed in each nontrivial iteration in this order, for some $\ell \leq n^2$. In Step 3, since the Tatami printer passes for every iteration, $D_i$ does not overlap with $D_1,D_2,...,D_{i-1}$ for every $i$. Also, since the final verification passes, the combined area of $D_1,D_2,...,D_\ell$ must cover the whole puzzle grid, i.e. $D_1,D_2,...,D_\ell$ form a partition of the grid. Furthermore, each point in the grid can be a corner point of at most two squares, so no four squares share a corner point.

Since the final verification passes, all integers in the original puzzle match the ones on the corresponding cards.

Hence, we can conclude that the puzzle grid is partitioned into squares according to the rules, which implies $P$ must know a valid solution of the puzzle.
\end{proof}

\begin{lemma}[Zero-Knowledge] \label{lem6}
During the verification, $V$ obtains no information about $P$'s solution.
\end{lemma}

\begin{proof}
To prove the zero-knowledge property, we will show that any interaction between $P$ and $V$ can be simulated by a simulator $S$ that does not know $P$'s solution. It is sufficient to prove that all distributions of cards that are turned face-up can be simulated by $S$.

The zero-knowledge property of the chosen cut protocol and the Tatami printer has been proved in the proof of Lemma \ref{lem3}, so it is sufficient to consider only the main protocol.

In Step 4 of the main protocol, there is only one deterministic pattern of the cards that are turned face-up (all correct templates). Therefore, this step can be simulated by $S$.

Hence, we can conclude that $V$ obtains no information about $P$'s solution.
\end{proof}

\section{Future Work}
We developed the Tatami printer protocol, which can check the corner constraint of Tatami puzzles. However, this protocol still has a limitation that it cannot check other types of constraint involving relationships between multiple regions, e.g. a constraint in \textit{Fillmat} which specifies that horizontally or vertically adjacent rectangles cannot contain the same number. An interesting future work is to develop a technique to check such constraints.


\begin{thebibliography}{99}
	\bibitem{np} A. Adler, J. Bosboom, E.D. Demaine, M.L. Demaine, Q.C. Liu and J. Lynch. Tatamibari Is NP-Complete. In \textit{Proceedings of the 10th International Conference on Fun with Algorithms (FUN)}, pp. 1:1--1:24 (2020).
	\bibitem{zkp} O. Goldreich, S. Micali and A. Wigderson. Proofs that yield nothing but their validity and a methodology of cryptographic protocol design. \textit{Journal of the ACM}, 38(3): 691--729 (1991).
	\bibitem{zkp0} S. Goldwasser, S. Micali and C. Rackoff. The knowledge complexity of interactive proof systems. \textit{SIAM Journal on Computing}, 18(1): 186--208 (1989).
	\bibitem{sumplete} K. Hatsugai, S. Ruangwises, K. Asano and Y. Abe. NP-Completeness and Physical Zero-Knowledge Proofs for Sumplete, a Puzzle Generated by ChatGPT. \textit{New Generation Computing}, 42(3): 429--448 (2024).
	\bibitem{janko} A. Janko and O. Janko. Square Jam. \url{https://www.janko.at/Raetsel/Square-Jam/index.htm}
	\bibitem{janko2} A. Janko and O. Janko. Tatamibari. \url{https://www.janko.at/Raetsel/Tatamibari/index.htm}
	\bibitem{koch} A. Koch and S. Walzer. Foundations for Actively Secure Card-Based Cryptography. In \textit{Proceedings of the 10th International Conference on Fun with Algorithms (FUN)}, pp. 17:1--17:23 (2020).
	\bibitem{slitherlink} P. Lafourcade, D. Miyahara, T. Mizuki, L. Robert, T. Sasaki and H. Sone. How to construct physical zero-knowledge proofs for puzzles with a ``single loop'' condition. \textit{Theoretical Computer Science}, 888: 41--55 (2021).
	\bibitem{takuzu} D. Miyahara, L. Robert, P. Lafourcade, S. Takeshige, T. Mizuki, K. Shinagawa, A. Nagao and H. Sone. Card-Based ZKP Protocols for Takuzu and Juosan. In \textit{Proceedings of the 10th International Conference on Fun with Algorithms (FUN)}, pp. 20:1--20:21 (2020).
	\bibitem{kakuro} D. Miyahara, T. Sasaki, T. Mizuki and H. Sone. Card-Based Physical Zero-Knowledge Proof for Kakuro. \textit{IEICE Trans. Fundamentals}, E102.A(9): 1072--1078 (2019).
	\bibitem{puzz} puzz.link: List of puzzle types. \url{https://puzz.link/list.html}
	\bibitem{nurikabe} L. Robert, D. Miyahara, P. Lafourcade and T. Mizuki. Card-Based ZKP for Connectivity: Applications to Nurikabe, Hitori, and Heyawake. \textit{New Generation Computing}, 40(1): 149--171 (2022).
	\bibitem{nonogram2} S. Ruangwises. An Improved Physical ZKP for Nonogram and Nonogram Color. \textit{Journal of Combinatorial Optimization}, 45(5): 122 (2023).
	\bibitem{goishi} S. Ruangwises. Verifying the First Nonzero Term: Physical ZKPs for ABC End View, Goishi Hiroi, and Toichika. \textit{Journal of Combinatorial Optimization}, 47(4): 69 (2024).
	\bibitem{shikaku} S. Ruangwises and T. Itoh. How to Physically Verify a Rectangle in a Grid: A Physical ZKP for Shikaku. In \textit{Proceedings of the 11th International Conference on Fun with Algorithms (FUN)}, pp. 24:1--24:12 (2022).
	\bibitem{numberlink} S. Ruangwises and T. Itoh. Physical Zero-Knowledge Proof for Numberlink Puzzle and $k$ Vertex-Disjoint Paths Problem. \textit{New Generation Computing}, 39(1): 3--17 (2021).
	\bibitem{ripple} S. Ruangwises and T. Itoh. Physical Zero-Knowledge Proof for Ripple Effect. \textit{Theoretical Computer Science}, 895: 115--123 (2021).
	\bibitem{decom} S. Ruangwises and M. Iwamoto. Printing Protocol: Physical ZKPs for Decomposition Puzzles. \textit{New Generation Computing}, 42(3): 331--343 (2024).
	\bibitem{sudoku} T. Sasaki, D. Miyahara, T. Mizuki and H. Sone. Efficient card-based zero-knowledge proof for Sudoku. \textit{Theoretical Computer Science}, 839: 135--142 (2020).
	\bibitem{polygon} K. Shinagawa, T. Mizuki, J.C.N. Schuldt, K. Nuida, N. Kanayama, T. Nishide, G. Hanaoka and E. Okamoto. Card-Based Protocols Using Regular Polygon Cards. \textit{IEICE Trans. Fundamentals}, E100.A(9): 1900--1909 (2017).
\end{thebibliography}
\end{document}